\documentclass[pra,showpacs,superscriptaddress,floatfix]{revtex4-1}
\usepackage{amsfonts,amssymb,amsmath}
\usepackage[pagebackref=true,colorlinks,linkcolor=blue,citecolor=red]{hyperref}
\usepackage{graphicx}
\usepackage{amsthm}

\newtheorem*{theorem*}{Theorem}
\newtheorem*{lemma*}{Lemma}

\def\be{\begin{equation}}
\def\ee{\end{equation}}
\def\bea{\begin{eqnarray}}
\def\eea{\end{eqnarray}}
\def\la{\langle}
\def\ra{\rangle}

\begin{document}
\title{Thermal stability of two-dimensional Topological Color Code}

\author{Razieh Mohseninia}
\email{mohseninia@physics.sharif.ir}
\affiliation{Department of Physics, Sharif University of Technology, P.O. Box 11155-9161, Tehran, Iran.}

\begin{abstract}
Thermal stability of the Topological Color Code in presence of a thermal bath is studied. We study the Lindblad evolution of the observables in the weak coupling limit of the Born-Markov approximation. The  auto-correlation functions of the observables are used as a figure of merit for the thermal stability. We show that all of the observables auto-correlation functions decay exponentially in time. By finding a lower bound of the decay rate, which is a constant independent of the system size, we show that the Topological Color Code is unstable against thermal fluctuations from the bath at finite temperature, even though it is stable at T=0 against local quantum perturbations.
\end{abstract}
\pacs{03.67.-a, 03.65.Vf, 05.30.Pr, 03.65.Yz, 03.67.Pp, 02.50.Ga.}
\maketitle

\section{Introduction}\label{intro}
The fragility of the qubits in presence of decoherence and external noise is the biggest obstacle in realizing a scalable quantum computer. To overcome such problems, quantum error correcting codes have been invented \cite{error1,error2,error3,error4,error5,error6,error7,error8,error9}. The main idea of the error correcting codes is to encode  information in a many particle system, this many particle system plays the role of a stable  logical qubit. However, error correcting models are themselves cause of errors and the error threshold below which one can perform fault tolerant quantum computation is very low \cite{preskillthreshold, barbara terhal}.

Topological quantum codes have emerged as the most promising candidates to achieve fault tolerant quantum computation. In these models information is stored in global properties of the model, say topologically degenerate ground states of the system. In particular, this kind of coding is shown to be robust against local perturbations from the environment provided they are local in space and time and they occur at T=0 temperature. Examples of good topological codes are Kitaev Code \cite{kitaev} and Topological Color Code (TCC) \cite{color code1}.  A very important figure of merit to assess the goodness of a topological code is the error threshold. When only qubit errors occur, the error threshold for the TCC turns out to match the one by the Kitaev model, namely, 11$\%$ \cite{threshold1,threshold2}. However, for more realistic situations when the measurement process is  also prone to errors, the TCC threshold is 4.5 $\%$ \cite{threshcolor1,threshcolor2}, even better than 
the one for the Kitaev Code which is 2.9 $\%$ \cite{threshkitaev}.

Topological Color Codes have shown very versatile properties for doing fault-tolerant quantum computation. In 2D, a TCC can implement the Clifford group of gates in a transversal way \cite{color code1}; this implementation of Clifford group with TCC makes quantum teleportation, distillation of entanglement and dense coding possible in a fully topological manner. Moreover, three dimensional extensions of TCC can also achieve universal quantum computation \cite{color code2}. The first realization of this model has been done in \citep{science}.

An open problem is to find topological codes resilient to thermal fluctuations from the environment. A first indication that the behavior of topological codes may be different at non-zero temperature was advanced in \cite{threshold1, threshold2}, and then it was confirmed by a rigorous proof in \cite{Horo,Horodecki2} within the setting of the dynamics of quantum open systems governed by Lindblad dynamics. It has been shown in \cite{Horo,Horodecki2} that the Kitaev model in two spacial dimensions is not a stable memory in presence of a thermal bath. Interestingly enough, it is possible to stabilize topological codes under thermal noise provided the lattice system can be defined in higher spatial dimensions. Namely, a thermally robust topological quantum memory in D=4 spatial dimensions can be constructed with the Kitaev code \cite{Horodecki4}, and a fully fledged universal quantum computer robust to thermal noise can be constructed in D=6 dimensions with Topological Color Codes \cite{selfcolor}.

In this paper we address the problem of thermal stability of TCC in a two-dimensional lattice, based on a mathematically rigorous analysis of the thermal effects on the model. According to the paper \cite{kubica}, the Color Code on a two dimensional hexagonal lattice can be mapped to two de-coupled Kitaev models on two dimensional triangular lattices by local unitary actions. In the original lattice on which the Color Code is defined (Hexagonal lattice) qubits lie on vertices of the lattice, however on the mapped model the qubits lie on the edges of the triangular lattices. Although the Color Code is mapped to two de-coupled Toric codes, an error applied from the bath on a single qubit of the Color Code, corresponds in the mapped model, to errors causing excitations in the two de-coupled Toric codes, i.e. the two disjoined lattices. This means that the processes of creation of excitations in these two disjoined lattices are not independent from each other, thus, due to the coupling to the bath, these two disjoined Toric codes can be correlated. This possibility was not taken into account in the previous works on the stability of Kitaev model, and therefore, by knowing the thermal stability properties of one Toric code one can not gain any information about the thermal stability of the Color Code, and this problem is not trivial. 

The method that we use is similar to the one that is used in \citep{Horodecki2}. To this end, we study  the dynamics of the TCC, weakly interacting with a heat bath in the Born-Markov approximation. The evolution of the observables governed by Lindblad dynamics and their auto-correlation function in time is studied as a tool for proving the instability of this model. We show that all of the observables auto-correlation functions decay exponentially with a constant decay rate which means that the model in unstable against thermal noise, although it is stable against local quantum perturbation at zero temperature.

The rest of the paper is organized as follows: in Sec. \ref{color code} we review the main features of the TCC. In Sec. \ref{Markov} we provide some basic results of the Markovian approximation in the weak coupling limit. Sec. \ref{s&ins} deals with reviewing the stability and instability conditions of the topological memories. In Sec. \ref{instability} these conditions are checked for the case of TCC and its instability is proved rigorously. Finally Sec. \ref{conclusion} is devoted to conclusive remarks. In appendix A, we prove the negativity of Lindblad super-operator.

\section{Topological Color Code}\label{color code}
Topological Color Code is a class of topological codes that can be defined on any three colorable lattice, where by colorable we mean colorable by face or equivalently by edge \cite{color code1}. In the present work, we consider a two-dimensional hexagonal lattice, 2-colex \cite{color code3}, on which the periodic boundary conditions are imposed on both sides. This lattice consists of $N$ plaquettes, $2N$ vertices and $3N$ edges. A three colorable lattice is a lattice on which one can color its plaquettes with three different colors (Red, Green, Blue\cite{foot1}) in a way that the plaquettes with the same color do not share any links. Each link connects two plaquettes with the same color and, therefore, one can ascribe every link with this special color.

The qubits, in this model, live on the vertices. The Hamiltonian of the model consists of two kinds of plaquette operators, $B_{p}^x$ and $B_{p}^z$, which are defined as follows:
\be
B_{p}^x=\prod_{i \in p} \sigma_{x,i}, \qquad \qquad \qquad B_{p}^z=\prod_{i \in p} \sigma_{z,i},
\ee
where $\sigma_{x}$ and $\sigma_z$ are ordinary Pauli operators and $p$ denotes a plaquette. Note that $B_p^x$ and $B_p^z$ can be defined for all the plaquettes and thus we have a total of $N$ distinct $B_p^x$ operators and $N$ distinct $B_p^z$ operators. The Hamiltonian is given by:
\be
H=-J\sum_pB_{p}^x-J\sum_pB_{p}^z,
\ee
where, the summation is done over all of the plaquettes. All of the operators in the Hamiltonian commute with each other, since they either share two qubits or none. Thus, the Hamiltonian is exactly solvable. The plaquette operators also square to identity and therefore, have $\pm1$ eigenvalues.

One should note that there are $2N$ qubits and $2N$ stabilizers \cite{foot2} (for further details about stabilizer quantum codes please see \cite{stab1} and \cite{stab2}) in the Hamiltonian. Nevertheless, all of these stabilizers are not independent, because of these constrains on the torus:
\be \label{cons}
\prod_{p \in R} B_{p}^\sigma =\prod_{p \in B} B_{p}^\sigma =\prod_{p \in G} B_{p}^\sigma, \qquad \qquad \sigma=x,z.
\ee

\begin{figure}[t]
\begin{center}
\includegraphics[scale=0.40]{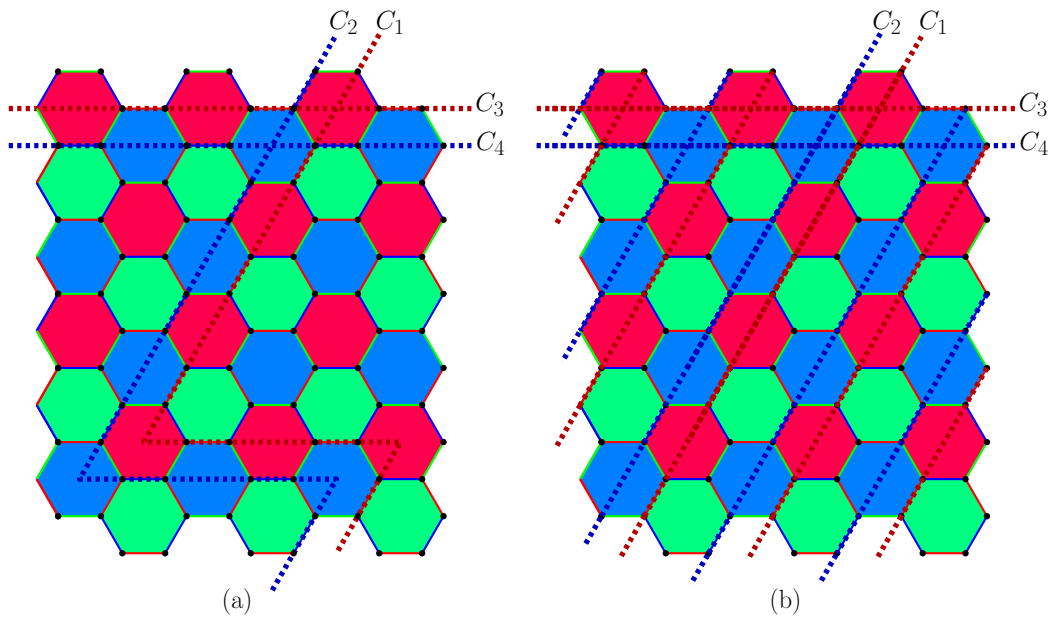}
\caption{(Color Online): TCC on a hexagonal lattice. The logical operators are defined on non-trivial loops shown with $C_i$'s. (a) The non-trivial loops $C_1$ and $C_2$ turn around the torus once. (b) The non-trivial loops $C_1$ and $C_2$ turn around the torus three times and they pass all of the plaquettes having the same color.}\label{lattice}
\end{center}
\end{figure}

The number of constrains for each type of plaquette operators ($\lbrace B_p^x \rbrace$ and $\lbrace B_p^z \rbrace$) is 2, therefore, there are $2N-4$ independent stabilizers in the Hamiltonian and the Hamiltonian has $\frac{2^{2N}}{2^{2N-4}}=16$ degenerate ground states. The ground sub-space of the Hamiltonian is the sub-space spanned by the states which are stabilized by all of the plaquette operators simultaneously ($B_{p}^x |gs\ra=B_{p}^z|gs\ra=|gs\ra$) and one of these states can be represented as:
\be
|gs\ra=\prod_p \left(1+B_{p}^x \right)|0\ra^{\otimes N},
\ee
up to a normalization factor. To construct the other $15$ ground states one needs to define the following logical operators:
\be \label{logicZ}
{Z}_1=\prod_{i \in C_1} \sigma_{z,i}, \quad  {Z}_2=\prod_{i \in C_2} \sigma_{z,i}, \quad {Z}_3=\prod_{i \in C_3} \sigma_{z,i},  \quad {Z}_4=\prod_{i \in C_4} \sigma_{z,i},
\ee
\noindent
\be \label{logicX}
{X}_1=\prod_{i \in C_4} \sigma_{x,i},  \quad {X}_2=\prod_{i \in C_3} \sigma_{x,i},  \quad {X}_3=\prod_{i \in C_2} \sigma_{x,i}, \quad {X}_4=\prod_{i \in C_1} \sigma_{x,i},
\ee 
where $C_1$, $C_2$, $C_3$ and $C_4$ are four non-trivial loops in the torus in the sense that they can not be written as a tensor product of some plaquette operators (figure \ref{lattice}). One should note that there are only two non-trivial loops for each non-trivial homology cycles in a torus, the Blue loop and the Red loop; the third non-trivial loop (Green) can be written as a tensor product of the red and the blue ones, i.e. 
\be \label{equivalence}
C_r C_b C_g \sim 1, \qquad \qquad \qquad  C_r C_b \sim C_g,
\ee
up to some plaquette operators. Using these logical operators, all of the $16$ ground states can be represented as follows:
\be
|i_1,i_2,i_3,i_4\ra = {X}_{1}^{i_1} {X}_{2}^{i_2} {X}_{3}^{i_3} {X}_{4}^{i_4} |gs\ra, \qquad \qquad i_n=0,1, \qquad n=1,2,3,4.
\ee

The non-trivial loops can be represented in two ways as shown in figure \ref{lattice}. In figure \ref{lattice}a they turn around the torus once, while in figure \ref{lattice}b they turn around three times (which is a function of the system size). One should note that the two types of representing the loops are equivalent in the sense that one can deform the two representations into each other by using a set of appropriate plaquette operators. The second representation (figure \ref{lattice}b) will be used in section \ref{instability}.

In a realization of Topological Color Code, information can be stored in the topologically degenerate ground states of the system. One can use the ground states to encode $4$ logical qubits. Due to its topological order, the model is robust against local perturbations and the only perturbations that may cause logical error are those with a length equal to the system size. Moreover the Clifford group's generators can be implemented by this model, which is sufficient for doing quantum distillation of entanglement without any need to address single qubits and to braid the quasi-particles \cite{color code1,color code2}.

\section{Markovian approximation in the weak coupling limit}\label{Markov}
Consider a quantum system which is not closed  and is coupled to a thermal bath at temperature $T$. One can attribute the following total Hamiltonian to the system and the bath, which form a closed system together:
\be
H=H^{\text{sys}}+H^{\text{bath}}+H^{\text{int}}, \qquad \qquad H^{\text{int}}= \sum_\alpha S_\alpha \otimes f_\alpha,
\ee
where, $H^{\text{sys}}$ is the topologically ordered Hamiltonian of the system whose stability is being studied and $H^{\text{bath}}$ is the Hamiltonian of the bath that we do not have any knowledge of it and $H^{\text{int}}$ is the system-bath interaction Hamiltonian. $S_\alpha$'s are operators acting on the system and $f_\alpha$'s are operators acting on the bath, and without loss of generality we can assume that they are Hermitian \cite{angel rivas}.\\

In the weak coupling limit of the interaction Hamiltonian, an operator in the Heisenberg picture, evolves as follows \cite{angel rivas,week1,week2,week3}:
\be
\frac{dX}{dt}= \mathcal{G} (X):= i[H^{\text{sys}},X]+ \mathcal{L}(X).
\ee
where, $i$ is the complex imaginary unit and $\mathcal{G}$ is the generator of the evolution, which consists of two parts. The first part is the normal generator of the evolution of closed quantum systems and the second one is the Lindblad generator or the dissipative part of the evolution, due to existence of the bath. The latter can be given as \cite{angel rivas,week1,week2,week3}:
\bea \label{lindblad}
\mathcal{L}(X)&=&  \sum_\alpha \sum_{\omega \geq 0} \mathcal{L}_{\alpha,\omega} (X)\\ \nonumber
&=& \frac{1}{2} \sum_\alpha \sum_{\omega \geq 0} h_\alpha(\omega) \Bigg( S^\dagger_\alpha(\omega) [X, S_\alpha(\omega)]+[S^\dagger_\alpha(\omega) ,X] S_\alpha(\omega) \\ \nonumber
&+& e^{- \beta \omega} S_\alpha(\omega) [X, S^\dagger_\alpha(\omega)] +e^{- \beta \omega} [S_\alpha(\omega) ,X] S^\dagger_\alpha(\omega)\Bigg),  \nonumber
\eea
where $\beta$ is the inverse of the temperature of the system and the factors $h_\alpha(\omega)$'s are the Fourier transforms of the auto-correlation functions of $f_\alpha$'s and we have used the relations $h_\alpha(-\omega)=e^{- \beta \omega} h_\alpha(\omega)$. In addition $S_\alpha(\omega)$ is the Fourier transform of  $S_\alpha$:
\be
S_\alpha (\omega) =\sum_{\epsilon-\epsilon^\prime=\omega} \Pi_{\epsilon^\prime}  S_\alpha  \Pi_{\epsilon},
\ee
where $\Pi_{\epsilon}$  is the projector onto the sub-space with energy $\epsilon$ and $\omega$'s are the Bohr frequencies of the system Hamiltonian. One can further check that $S_\alpha(-\omega)=S^\dagger_\alpha(\omega)$ and $\sum_\omega S_\alpha(\omega)=S_\alpha$.

\subsubsection*{\textbf{Properties of the Lindblad super-operator}}
In this section we briefly review some of the essential features of  the Lindblad super-operator, needed for our study:
\begin{itemize}
\item  \textbf{Self-adjointness of $\mathcal{L}$ }: If we define the Liouville scalar product as follows:
\be
\la X,Y \ra_\beta:= \text{tr} (\rho_\beta X^\dagger Y),
\ee
the Lindblad super-operator is self-adjoint with respect to it, i.e:
\be
\la X,\mathcal{L}(Y) \ra_\beta=\la \mathcal{L}(X),Y \ra_\beta.
\ee
From here on, in the rest of this paper by scalar product we mean the Liouville scalar product and we withdraw writing $\beta$ symbol.
\item  \textbf{Positivity of $- \mathcal{L}$ }: The Lindblad super-operator is negative which means that:
\be
-\la X, \mathcal{L}(X)\ra \geq 0, \qquad  \forall  X.
\ee
The negativity of $\mathcal{L}$ is proved in the Appendix.
\item \textbf{Gap of $\mathcal{-L}$}: Because of the positivity of $-\mathcal{L}$, its smallest eigenvalue different from $0$ is defined as its gap:
\be
\mathrm{Gap} (-\mathcal{L}) := \min_X {\bigg(- \la X, \mathcal{L}(X)\ra:  \forall X \neq I \bigg)}.
\ee
where, $I$ is the identity operator.
\end{itemize}

\section{Stability and instability conditions for topological memories}\label{s&ins}
\begin{itemize}
\item \textbf{Stability}: To prove the stability of a memory at finite temperature and its capability for coding the information, one should find an observable as the logical operator for the logical qubit such that by increasing the system size the auto-correlation function of the observable does not decrease in time. More rigorously one should find an observable $X$ and a decay rate $\epsilon$, such that:
\be \label{sta}
\la X,X(t) \ra \geq e^{- \epsilon t} \la X,X \ra,
\ee
where $\epsilon$ is the decay rate of the auto-correlation function of the observable $X$. In case of a stable memory the decay rate should decrease exponentially with system size ($\epsilon=e^{-N a}$), so that by increasing the system size the decay rate goes to zero. This means that the autocorrelation function of the an observable, in the limit of large system size, will not decrease in time and the memory will be stable and self-correcting \cite{Horodecki4}. By substituting  $X(t)= e^{t \mathcal{L} }X$ into equation \ref{sta}, the condition for the stability recasts into the following:
\be
- \la X, \mathcal{L}(X)\ra \leq \epsilon.
\ee
where $\epsilon$ decays exponentially with the system size.

\item  \textbf{Instability}: To prove the instability of a memory one should prove that the auto-correlation function of all of the observables with time, decreases faster than an exponential function, which means that for any observable we have:
\be
\la X,X(t) \ra \leq e^{- \epsilon t}  \la X,X \ra.
\ee
By substituting $X(t)= e^{t \mathcal{L} }X$ into the above equation the instability condition of a memory recasts into the following:
\be
- \la X, \mathcal{L}(X)\ra \geq \epsilon,
\ee
which means that for proving the instability of a memory one should estimate a lower bound of $- \la X, \mathcal{L}(X)\ra$ by minimizing it over all of the observables. If this quantity is a constant independent of the system size or is a variable of the system size that does not decrease with the size of the system, the memory is unstable. Since in the finite time the auto-correlation goes to zero and the encoded information lost (for more details please see \cite{Horodecki2,analysis}).\\

Therefore, proving the stability of a memory is nothing but obtaining the gap of $-\mathcal{L}$, which is denoted by $G$:
\be
\mathrm{Gap} (-\mathcal{L}) := G=\min_X  \bigg(- \la X, \mathcal{L}(X)\ra : \forall X \neq I \bigg).
\ee
Applying Eq. (\ref{positive}) (see Appendix A for further details), one obtains:
\be
G \geq  \frac{1}{2}\min_{\alpha, \omega}(h_\alpha(\omega))\min_X{\bigg( \sum_{\alpha,\omega} \la [S_\alpha(\omega),X],[S_\alpha(\omega),X]\ra\bigg)}.
\ee
Moreover, by using equation 28 of reference \cite{Horodecki4} which indicates that:
\be
\sum_\omega \la [S_\alpha(\omega),X],[S_\alpha(\omega),X]\ra = \sum_{\omega, \omega^\prime} \la [S_\alpha(\omega),X],[S_\alpha(\omega^\prime),X]\ra
\ee

and the relation $\sum_\omega S_\alpha (\omega)=S_\alpha$, we obtain:
 \be
G \geq \frac{1}{2}\min_{\alpha, \omega}(h_\alpha(\omega)) \min_X{ \bigg(\sum_{\alpha} \la [S_\alpha,X],[S_\alpha,X]\ra \bigg)},
\ee
which is easier to estimate. One should note that $h_\alpha(\omega)$ is the Fourier transform of the auto-correlation function of $f_\alpha$  and it can be supposed that $h_\alpha(\omega)$ does not depend on $\alpha$'s, which means that the action of the bath and the strength of interaction Hamiltonian is uniform in the whole system. Thus, the minimum of $h_\alpha(\omega)$ is equal to $e^{- \beta \Delta} h(\Delta)$, where $\Delta$ is the gap of the Hamiltonian and for TCC, $\Delta$ is equal to $6J$. Therefore, the lower bound of the gap recasts into the following:
\be \label{lower bound}
G \geq \frac{1}{2} e^{- \beta \Delta} h(\Delta) \min_X{ \bigg(\sum_{\alpha} \la [S_\alpha,X],[S_\alpha,X]\ra \bigg)}.
\ee
\end{itemize}

\section{Thermal instability of the Topological Color Code} \label{instability}
Consider a realization of the Topological Color Code, which is coupled to a thermal bath at temperature $T$. Due to this coupling, errors can be applied from the bath on the system. The errors usually do not commute with the Hamiltonian and the system will not remain in the ground states anymore. If the system is able to correct itself, i.e. it can remove errors to stay in the ground sub-space, it can be considered as a stable topological memory. \\

In the present work, thermal stability of TCC at finite temperature is studied. We assume that the interaction Hamiltonian between the system and the bath is of the following form:
\be \label{interaction}
H^{\text{int}}=\sum_i \left( \sigma_{x,i} \otimes f_{x,i} + \sigma_{z,i} \otimes f_{z,i}\right),
\ee
where $\sigma_{x}$ and $\sigma_{z}$ are applied from the bath on each qubit. To understand the effect of this Hamiltonian on the system, consider the $i^{th}$ qubit, for example; $\sigma_{x,i}$ ($\sigma_{z,i}$) anti-commutes with the three $z$ type ($x$ type) plaquette operators that have this qubit in common. Thus, if $\sigma_{x,i}$ acts from the bath on the ground state, because of this anti-commutation, the eigenvalues of the three paluettes operators become $-1$ and the system leaves the ground sub-space and consequently the code space. 

We can assume that in a plaquette with $-1$ eigenvalue, an excitation (quasi-particle) has been created. In TCC the excitations can move freely and cause logical errors. Therefore, it seems that the model is not self correcting. By having this intuition we present a rigorous proof of the instability of this model, i.e we shall estimate a lower bound for the gap of the Lindblad super-operator corresponding to the model as it is discussed in the previous section.

\subsection{Excitations}\label{Exc}
If the operators that are applied from the bath on the system do not commute with the Hamiltonian, they create excitations. In this section we introduce the generators for having all possible excitations in TCC. In this model the excitations appear in many different ways and not necessarily in pairs; however, all of them can be generated by the use of two kinds of generators:
\begin{itemize}
\item \textbf{Open strings}: Corresponding to each color, there is a global open string as shown in figure \ref{excitations}b. The red global open string, for example, is obtained by inserting a site at the center of every red plaquette and then connecting the sites through the red links which have the same orientation \cite{foot3} (exactly as the $C_r$ string shown in figure \ref{excitations}b). Note that each link of the global strings corresponds to two nearest neighbor qubits in the original lattice.
Let us call the connected subsets of the global open strings as open strings.

 As an example, consider an open string with only one link. By acting with $\sigma_x  \otimes\sigma_x$ or $\sigma_z  \otimes \sigma_z$ on the two qubits that lie on the only link of the string, the two plaquettes that are on the two ends of the string, will be excited (figure \ref{excitations}a). By increasing the length of this string one can move the excitations to any other two red plaquettes.

In this case, excitations appear in pairs. In a lattice with $N$ plaquettes and $2N$ qubits, to generate all of the open string operators, $2N-6$ qubits are needed; these qubits are located on $C_r, C_b$ and $C_g$ global strings.

\begin{figure}[t]
\begin{center}
\includegraphics[scale=0.35]{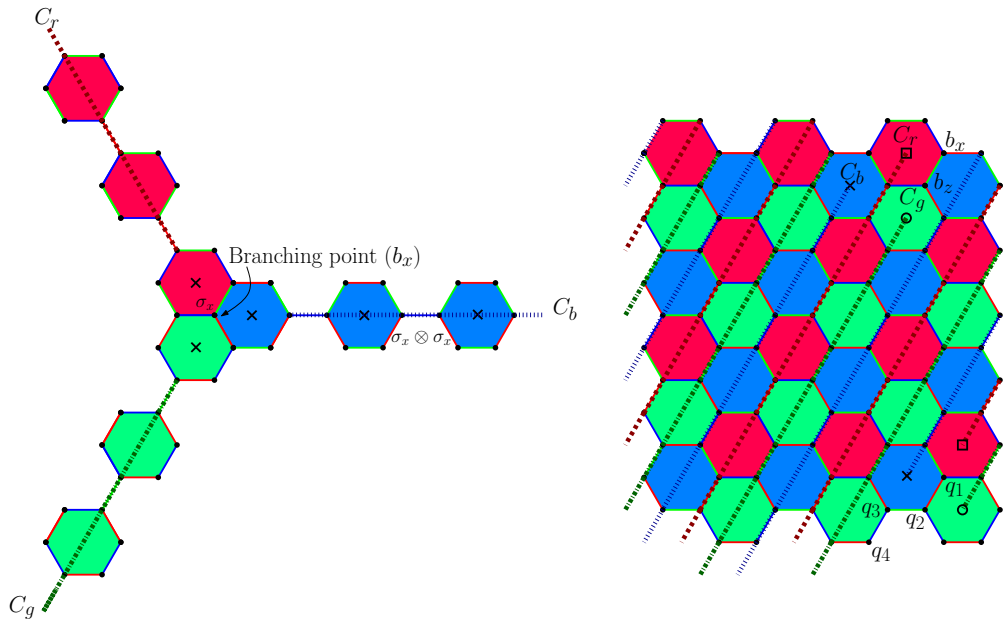}
\caption{(Color Online): The open strings $C_r, C_b, C_g$ and branching points $b_x, b_z$ for creating all kinds of excitations in TCC.}\label{excitations}
\end{center}
\end{figure}

\item \textbf{Branching points}: In TCC it is possible to have three excitations in three plaquettes with different colors. It is impossible to create such excitations by using the open strings; however by using a single-qubit operator along with the open string operators, one can have all kinds of excitations. More rigorously, by acting with $\sigma_x$ on one qubit, which we call branching point, excitations are created in the three plaquettes that have this qubit in common \cite{foot4}. By using open strings one can move the excitations from these three plaquettes to other plaquettes with the same color.

One should note that because of the relation \ref{cons}, arbitrary number of excitations for different colors are not allowed and there are certain constraints on the number of excitations of different colors. For example, a single excitation with red color is not allowed. 
\end{itemize}
Therefore, all kinds of excitations can be generated by applying $\sigma_{x}\otimes \sigma_{x}$ or $\sigma_{z}\otimes \sigma_z$ on the qubits that belong to the open strings and also $\sigma_x$ and $\sigma_z$ on the branching points $b_x$ and $b_z$, respectively.

In a hexagonal lattice with $2N$ qubits, to generate all kinds of excitations one needs $2N-6+2=2N-4$ qubits. Thus, four qubits are left ($q_1, q_2, q_3$ and $q_4$ shown in figure \ref{excitations}b). This is consistent with having 16 degenerate ground states and 4 logical qubits.

\subsection{Observables} \label{Obs}
All of the observables corresponding to a 2-dimensional Hilbert Space, can be generated by $\sigma_x$ and $\sigma_z$. Therefore, the algebra of the observables for a system consisting of $N$ qubits may be written as:
\be
\mathcal{O}=\mathcal{Q}_1 \otimes \mathcal{Q}_2 \otimes \dots \otimes \mathcal{Q}_N,
\ee
where $\mathcal{Q}_i$ is the algebra of the observables of the $i^{th}$ qubit which is generated by $\sigma_{x,i}$ and $\sigma_{z,i}$. However, one can construct all of the observables in another way by the use of the operators present in the Hamiltonian, the logical operators and  the generators needed to create all kinds of excitations. The latter depends on the form of the Hamiltonian. For the case of TCC, the generators of the algebra of the observable are of the following three types:
\begin{itemize}
\item The $x$ and $z$ type plaquette operators in the Hamiltonian, $B_p^x$, $B_p^z$ operators.
\item The logical operators defined in equations \ref{logicZ} and \ref{logicX} \cite{foot5}.
\item The generators of the excitations:
\begin{itemize}
\item $\sigma_{x} \otimes \sigma_x$  or $\sigma_{z} \otimes \sigma_z$ acting on two nearest neighbor qubits that lie on $C_r$, $C_b$ and $C_g$ strings (figure \ref{excitations}).
\item $\sigma_x$ acting on $b_x$ and $\sigma_z$ acting on $b_z$.
\end{itemize}
\end{itemize}

\subsection{The Gap of generator of the Topological Color Code}
In this section, to prove the thermal instability of the TCC, we obtain a lower bound for the gap of generator of the evolution, due to the coupling to the thermal bath, and show that it is a constant independent of the system size. To this end we refer to equation \ref{lower bound} which for the interaction Hamiltonian defined in equation \ref{interaction} recasts into the following form:
 \be \label{lowergap}
G \geq \frac{1}{2} e^{- \beta \Delta} h(\Delta) \min_\Gamma{ \bigg(\sum_{i} \la [\sigma_{x,i},\Gamma],[\sigma_{x,i},\Gamma]\ra+  \la [\sigma_{z,i},\Gamma],[\sigma_{z,i},\Gamma]\ra\bigg)}.
\ee
The minimization in Eq. (\ref{lowergap}) is performed over all of the observables. However, if one wants an observable to be a logical observable acting on the code space, it should commute with the Hamiltonian. Therefore, we do not need to do the minimization over all of the observables on Hilbert space explained in the previous section; the observables that commute with the Hamiltonian would suffies. One can restrict the domain of the minimization even more; the logical observables $Z$ and $X$ for one logical qubit should anti-commute and square to identity. Therefore, all of the observables of our interest belong to the following algebra:
 \be \label{obs1}
\mathcal{O} = \left(Z^{\mu_1}_1X^{\nu_1}_1\right)( Z^{\mu_2}_2X^{\nu_2}_2)( Z^{\mu_3}_3X^{\nu_3}_3) (Z^{\mu_4}_4X^{\nu_4}_4)(\mathcal{B}_p^z\otimes\mathcal{B}_p^x), \qquad \qquad \mu_i=0,1 \text \ {and} \ \nu_i=0,1,
\ee
where $\mathcal{B}_p^z$ and $\mathcal{B}_p^x$ are the algebras generated by all of the $z$ and $x$ type plaquette operators respectively and the minimization is over different possibilities of $\mu_i$'s and $\nu_i$'s and also the two algebras $\mathcal{B}_p^x$ and $\mathcal{B}_p^x$. By putting an observable $\Gamma \in \mathcal O$ (some observable of our interest), into equation \ref{lowergap}, one finds that:
\be
G \geq  \frac{1}{2} e^{- \beta \Delta} h(\Delta)\left( \min_{\Gamma_z} \big( \sum_{i} \la [\sigma_{x,i},\Gamma_z],[\sigma_{x,i},\Gamma_z]\ra\big) + \min_{\Gamma_x}  \big(\sum_{i} \la [\sigma_{z,i},\Gamma_x],[\sigma_{z,i},\Gamma_x]\ra\big) \right).
\ee
Here by $\Gamma_z$ and $\Gamma_x$ we mean operators belonging to the following sub-algebras, respectively:
\be
\mathcal{O_z} = (Z^{\mu_1}_1 Z^{\mu_2}_2 Z^{\mu_3}_3 Z^{\mu_4}_4) \ \mathcal{B}_p^z,
\ee
\be
\mathcal{O_x} = (X^{\nu_1}_1 X^{\nu_2}_2  X^{\nu_3}_3  X^{\nu_4}_4) \ \mathcal{B}_p^x.
\ee
Therefore, we have:
\be
G \geq \frac{1}{2} (g_x+g_z),
\ee
Where, $g_x$ ($g_z$) comes from $\sigma_{x}$ ($\sigma_{z}$) part of the interaction Hamiltonian:
\be \label{gx}
g_x = e^{- \beta \Delta} h(\Delta) \min_{\Gamma_z} \bigg(\sum_{i} \la [\sigma_{x,i},\Gamma_z],[\sigma_{x,i},\Gamma_z]\ra\bigg),
\ee
\be
g_z = e^{- \beta \Delta} h(\Delta) \min_{\Gamma_x}{ \bigg(\sum_{i}   \la [\sigma_{z,i},\Gamma_x],[\sigma_{z,i},\Gamma_x]\ra\bigg)}.
\ee
By symmetry we know that $g_x$ and $g_z$ are equal to each other. Thus, the lower bound of the gap reduces to $G \geq g_x$, i.e:
\be
 G \geq  e^{- \beta \Delta} h(\Delta) \min_{\Gamma_z} \bigg(\sum_{i} \la [\sigma_{x,i},\Gamma_z],[\sigma_{x,i},\Gamma_z]\ra\bigg).
\ee
To obtain a lower bound of the gap, one can partition the algebra of all of the observables into different sectors (different possibilities of $\mu_i$'s) and obtain a lower bound for each sector and at the end, perform the minimization over all of the sectors.

\subsubsection{\textbf{The sector of $\Gamma_z=Z_1 \mathcal{B}_p^z$ or $\Gamma_z=Z_2\mathcal{B}_p^z$}}
Let us show the minimum of the decay rate in this sector by $G_1$. By symmetry, we know that the effect of the external bath on $Z_1 \mathcal{B}_p^z$ is exactly the same as its effect on $Z_2\mathcal{B}_p^z$. Thus, it is enough to consider only one of them, say  $Z_1 \mathcal{B}_p^z$:
\be
G_1 \geq e^{- \beta \Delta} h(\Delta) \min_{\Gamma_z}{ \bigg(\sum_{i} \la [\sigma_{x,i},Z_1 \mathcal{B}_p^z],[\sigma_{x,i},Z_1 \mathcal{B}_p^z]\ra\bigg)}.
\ee
From now on, we use the notation used in \cite{Horodecki2}, and show the terms $\la [\sigma_{x,i},Z_1 \mathcal{B}_p^z],[\sigma_{x,i},Z_1 \mathcal{B}_p^z]\ra$ as $\mathcal{L}_{x,i}^1$ for example. Therefore, we have:
\bea  \label{qwe}
\sum_{i} \la [\sigma_{x,i},Z_1 \mathcal{B}_p^z],[\sigma_{x,i},Z_1 \mathcal{B}_p^z]\ra= \sum_i \mathcal{L}_{x,i}^1 &=&  \mathcal{L}_{x,q_1}^1+ \mathcal{L}_{x,q_2}^1+ \mathcal{L}_{x,q_3}^1+ \mathcal{L}_{x,q_4}^1+\mathcal{L}_{x,b_x}^1 \\ \nonumber
 &+& \mathcal{L}_{x,b_z}^1+ \mathcal{L}_{x,C_b}^1+ \mathcal{L}_{x,C_g}^1+ \mathcal{L}_{x,C_r}^1\\ 
 &\geq & \mathcal{L}_{x,C_b}^1+ \mathcal{L}_{x,C_g}^1+ \mathcal{L}_{x,C_r}^1 \nonumber,
\eea
where by $\mathcal{L}_{x,C_b}^1$ we mean $ \sum_{i \in C_b} \mathcal{L}_{x,i}^1$ and so on. To obtain the inequality in Eq. (\ref{qwe}), we have used the positivity of each term $\mathcal{L}_{x,i}$, which is proved in Appendix~\ref{appendix a}. Now we calculate the terms $\mathcal{L}_{x,C_b}, \mathcal{L}_{x,C_g}$ and $\mathcal{L}_{x,C_r}$ separately as follows:
\bea
\mathcal{L}_{x,C_b}^1&=& \sum_{i \in C_b} \la [\sigma_{x,i},Z_1\mathcal{B}_p^z],[\sigma_{x,i},Z_1 \mathcal{B}_p^z]\ra \\ \nonumber
&=& \sum_{i \in C_b}\la Z_1 [\sigma_{x,i}, \mathcal{B}_p^z],Z_1[\sigma_{x,i}, \mathcal{B}_p^z]\ra \\ \nonumber
&=& \sum_{i \in C_b} \text{tr}\bigg( \rho_\beta [\sigma_{x,i}, \mathcal{B}_p^z]^\dagger Z_1^\dagger Z_1 [\sigma_{x,i}, \mathcal{B}_p^z]\bigg) \\ \nonumber
&=& \sum_{i \in C_b}\la [\sigma_{x,i}, \mathcal{B}_p^z],[\sigma_{x,i}, \mathcal{B}_p^z]\ra,
\eea
where the second line is the consequence of the fact that $Z_1$ commutes with all of $\sigma_{x,i}, i \in C_b$ (figure \ref{lattice}b and figure \ref{excitations}b), because the support of $Z_1$ and $C_b$ do not have any common qubit.\\

One can further check that the same arguments as above also hold for $\mathcal{L}_{x,C_g}^1$. On the contrary, the term $\mathcal{L}_{x,C_r}^1$leads to a different result, since the support of $Z_1$ and  $C_r$  meet each other and therefore, $Z_1$ does not commute with $\sigma_{x,i}$'s, $i \in C_r$. However, as shown in equation \ref{equivalence} in TCC a non-trivial loop with a specific color, say red, is equivalent to the tensor product of two other non-trivial loops that have different colors, green and blue, but are in the same homology class as the red one. Therefore, one can write $Z_1$, (denoted as $Z_r$) as the tensor product of $Z_2$ (denoted as $Z_b$) and another logical operator that is defined on a green non-trivial loop (denoted as $Z_g$), i.e.:
\be
Z_r Z_b Z_g \sim 1, \qquad \qquad \qquad Z_r \sim Z_g Z_b,
\ee
up to some plaquette operators. Thus,
\be
Z_r \mathcal{B}_p^z=Z_g Z_b \mathcal{B}_p^z,
\ee
where, we have absorbed the extra plaquette operators, in the algebra of all of the plaquette operators. Because $Z_g Z_b$ does not meet $C_r$ at any point, the same result as $\mathcal{L}_{x,C_{g,b}}^1$ also holds for the last term $\mathcal{L}_{x,C_r}^1$. Suppose that $A$ and $B$ have their minimum values at $X_1$ and $X_2$ respectively, since $A+B$ would in general, have its minimum at $X_3$ which is different from $X_1$ and $X_2$, one arrives at:
\be \label{ineq}
\min_{X} ( A+B) \geq \min_X(A) +\min_X(B). 
\ee
By using this inequality the lower bound of the gap of Lindblad in this sector recasts into the following:
\bea
G_1 \geq 3 e^{- \beta \Delta} h(\Delta) \min_ {\Gamma_z} \mathcal{L}_{x,C_{b}}^1.
\eea
Here, $\Gamma_z \in \mathcal{B}_p^z$. Obtaining the gap of this new model is simpler, because by knowing the effect of the bath on this new model one can map it to a known model that its Lindblad gap is known. The new model is nothing but the Ising model. The reason is that in TCC when $\sigma_x$ is applied on one qubit, say qubit number $1$ in figure \ref{Ising}a,  it can create three excitations in three plaquettes that have this qubit in common, by acting another $\sigma_x$ on the next qubit, qubit number $2$, two of these excitations will be annihilated and a new one can be created in the next blue plaquette (figure \ref{Ising}a). On the other hand consider another model, one dimensional Ising model with in-homogeneous couplings as follows:

\begin{figure}[t]
\begin{center}
\includegraphics[scale=0.45]{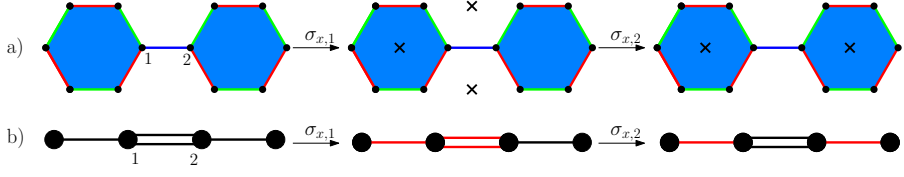}
\caption{(Color Online): Equivalence of the process of creation of excitation in the blue sub-lattice in the TCC , as shown in (a), and the in-homogeneous Ising model, as shown in (b).}\label{Ising}
\end{center}
\end{figure} 

\be \label{Is}
H_{\text{Ising}} = -J \sum_{i= \text{odd}} \sigma_{z,i} \sigma_{z,i+1}- 2J \sum_{i=\text{even}} \sigma_{z,i} \sigma_{z,i+1},
\ee

where $J$ is the coupling constant. In this model if $\sigma_x$ is applied to qubit number $1$ in figure \ref{Ising}b, it excites two of the bonds, however since one of the coupling strength is twice the other one, one can suppose that this external perturbation creates three excitations with the same energy (the excited bonds are shown with red color in figure \ref{Ising}b). By applying another $\sigma_x$ on qubit number $2$  the bond with two excitations is not excited any more, but another bond with one excitation can be excited. Thus, the process of creation and annihilation of excitation in the TCC is exactly what happens in the Ising model defined in equation \ref{Is}. It has been shown in \cite{Horodecki2} that the Ising model with arbitrary coupling is not a stable memory against thermal fluctuations, since the gap of Lindblad super-operator for this model is a constant independent of the system size. Therefore, one can conclude that for this sector of observables, the minimum of the decay rate is the following constant, which is independent of the system size:
\be
G_1 \geq 3  e^{- \beta \Delta} h(\Delta) (\text{Gap}(\mathcal{L}))_{\text{Ising}}.
\ee
\subsubsection{\textbf{The sector of $\Gamma_z=\mathcal{B}_p^z$}}

Let us show the minimum of the decay rate in this sector by $G_2$:
\be
G_2 \geq  e^{- \beta \Delta} h(\Delta) \min_{\Gamma_z}{ \bigg(\sum_{i} \la [\sigma_{x,i}, \mathcal{B}_p^z],[\sigma_{x,i}, \mathcal{B}_p^z]\ra\bigg)}.
\ee
As the previous case, the terms like $ \la [\sigma_{x,i}, \mathcal{B}_p^z],[\sigma_{x,i}, \mathcal{B}_p^z]\ra$ are shown as $\mathcal{L}_{x,i}^2$. Therefore, we have:
\bea \nonumber
\sum_{i} \la [\sigma_{x,i}, \mathcal{B}_p^z],[\sigma_{x,i}, \mathcal{B}_p^z]\ra= \sum_i \mathcal{L}_{x,i}^2 &=&  \mathcal{L}_{x,q_1}^2+ \mathcal{L}_{x,q_2}^2+ \mathcal{L}_{x,q_3}^2+ \mathcal{L}_{x,q_4}^2 \\ \nonumber
 &+& \mathcal{L}_{x,b_x}^2+ \mathcal{L}_{x,b_z}^2+ \mathcal{L}_{x,C_b}^2+ \mathcal{L}_{x,C_g}^2+ \mathcal{L}_{x,C_r}^2\\ 
 &\geq & \mathcal{L}_{x,C_b}^2+ \mathcal{L}_{x,C_g}^2+ \mathcal{L}_{x,C_r}^2.
\eea
By symmetry we know that the three terms $\mathcal{L}_{x,C_b}^2$, $\mathcal{L}_{x,C_g}^2$ and $\mathcal{L}_{x,C_r}^2$ are equal to each other. Thus, the the minimum decay rate recasts into the following form:
\be
G_2 \geq 3 e^{- \beta \Delta} h(\Delta) \min_{\Gamma_z} \mathcal{L}_{x,C_{b}}^2,
\ee
which is exactly what was discussed in previous section and was equal to the gap of Linddblad super-operator for the Ising model defined in equation \ref{Is}. Therefore, we have:
\be
G_2 \geq 3  e^{- \beta \Delta} h(\Delta) (\text{Gap}(\mathcal{L}))_{\text{Ising}}.
\ee

\subsubsection{\textbf{ The sector of $\Gamma_z=Z_3\mathcal{B}_p^z$ or $\Gamma_z=Z_4\mathcal{B}_p^z$}}
Let us show the minimum of the decay rate in this sector by $G_3$. By symmetry we know that the effect of the external bath on $Z_3 \mathcal{B}_p^z$ is exactly the same as the effect of the bath on  $Z_4\mathcal{B}_p^z$. Therefore, it is enough to consider only one of them, say  $Z_3 \mathcal{B}_p^z$:
\be
G_3 \geq e^{- \beta \Delta} h(\Delta) \min_{\Gamma_z}{ \bigg(\sum_{i} \la [\sigma_{x,i},Z_3 \mathcal{B}_p^z],[\sigma_{x,i},Z_3 \mathcal{B}_p^z]\ra\bigg)}.
\ee
In this sector the terms like $\la [\sigma_{x,i},Z_3 \mathcal{B}_p^z],[\sigma_{x,i},Z_3 \mathcal{B}_p^z]\ra$ are represented by $\mathcal{L}_{x,i}^3$. Therefore, we have:
\bea \label{eq:G3min} \nonumber
G_3 &\geq &  e^{- \beta \Delta} h(\Delta)\min_{\Gamma_z} \sum_i \mathcal{L}_{x,i}^3 = e^{- \beta \Delta} h(\Delta)\min_{\Gamma_z} \bigg( \mathcal{L}_{x,q_1}^3+ \mathcal{L}_{x,q_2}^3+ \mathcal{L}_{x,q_3}^3+ \mathcal{L}_{x,q_4}^3 \\ \nonumber
&+& \mathcal{L}_{x,b_x}^3+ \mathcal{L}_{x,b_z}^3+ \mathcal{L}_{x,C_b}^3+ \mathcal{L}_{x,C_g}^3+ \mathcal{L}_{x,C_r}^3 \bigg)\\ 
&\geq & e^{- \beta \Delta} h(\Delta) \min_{\Gamma_z} \bigg( \mathcal{L}_{x,C_b}^3+ \mathcal{L}_{x,C_g}^3+ \mathcal{L}_{x,C_r}^3 \bigg) \nonumber \\ 
&\geq& e^{- \beta \Delta} h(\Delta)\bigg( \min_{\Gamma_z} (\mathcal{L}_{x,C_b}^3)+ \min_{\Gamma_z}( \mathcal{L}_{x,C_g}^3)+  \min_{\Gamma_z}(\mathcal{L}_{x,C_r}^3) \bigg).
\eea
Now we calculate the terms separately as follows:
\be
 \min_{\Gamma_z} \mathcal{L}_{x,C_r}^3=  \min_{\Gamma_z}{ \bigg(\sum_{i \in C_r} \la [\sigma_{x,i},Z_3 \mathcal{B}_p^z],[\sigma_{x,i},Z_3 \mathcal{B}_p^z]\ra\bigg)}.
\ee
Since $Z_3$'s support is the closed red string and $C_r$ is also a red string, but in the other homological class, they do not meet each other at any point. Thus, $Z_3$ and $\sigma_{x,i}$, $i\in C_r$ commute with each other and the above equation reduces to:
\be 
\min_{\Gamma_z} \mathcal{L}_{x,C_r}^3=  \min_{\Gamma_z}{ \bigg(\sum_{i \in C_r} \la [\sigma_{x,i}, \mathcal{B}_p^z],[\sigma_{x,i}, \mathcal{B}_p^z]\ra\bigg)},
\ee
which is exactly what was discussed in the first case, and is equal to the gap of the Lindblad super-operator for the Ising model defined in equation \ref{Is}. The only quantities left to be obtained are $\mathcal{L}_{x,C_b}$ and  $\mathcal{L}_{x,C_g}$, which by symmetry we know to be equal to each other. Thus, it is sufficient to consider only one of them:
\be
\min_{\Gamma_z} \mathcal{L}_{x,C_{b}}^3  =  \min_{\Gamma_z}{ \bigg(\sum_{i \in C_{b}} \la [\sigma_{x,i},Z_3 \mathcal{B}_p^z],[\sigma_{x,i},Z_3 \mathcal{B}_p^z]\ra\bigg)}.
\ee
We can split the summation $\sum_{i \in C_{b}}$ in two parts, the qubits that lie on $Z_3$ loop and the qubits that do not lie on $Z_3$ loop:
\be
\mathcal{L}_{x,C_{b}}^3=  { \bigg(\sum_{i \in C_{b}, i\notin Z_3} \la [\sigma_{x,i},Z_3 \mathcal{B}_p^z],[\sigma_{x,i},Z_3 \mathcal{B}_p^z]\ra\bigg)} + { \bigg(\sum_{i \in C_{b}, i\in Z_3} \la [\sigma_{x,i},Z_3 \mathcal{B}_p^z],[\sigma_{x,i},Z_3 \mathcal{B}_p^z]\ra\bigg)}.
\ee
We use Lemma 2 in \cite{Horodecki2} which indicates that a lower bound of the gap of a super-operator like $\mathcal{S}$ that can be written as summation of two other super-operators, i.e. $\mathcal{S}=\mathcal{A}+\mathcal{B}$, is given by:
\be\label{56}
\text{Gap}(\mathcal{S})  \geq \frac{\text{Gap}(\mathcal{A}) \la \ker(\mathcal{A}), \mathcal{B}(\ker(\mathcal{A})) \ra}{\text{Gap}(\mathcal{A}) + ||\mathcal{B}||}.
\ee
We take $\mathcal{A}$ and $\mathcal{B}$ to be the Lindblad super-operator when the bath does not have any effect on the qubits lying on $C_3$ and when the bath is applied only on the qubits lying on $C_3$, respectively. Therefore, for the gap of $\mathcal{A}$ we have:
 \bea
\text{Gap}(\mathcal{A})&=&  \min_{\Gamma_z}{ \bigg(\sum_{i \in C_{b}, i\notin Z_3} \la [\sigma_{x,i},Z_3 \mathcal{B}_p^z],[\sigma_{x,i},Z_3 \mathcal{B}_p^z]\ra\bigg)} \\ \nonumber
 &=&  \min_{\Gamma_z}{ \bigg(\sum_{i \in C_{b}, i\notin Z_3} \la [\sigma_{x,i}, \mathcal{B}_p^z],[\sigma_{x,i}, \mathcal{B}_p^z]\ra\bigg)}.
\eea
One should note that the expression in the second line is not the gap of Ising model, because perturbations from the bath are not applied to all of the qubits lying on $C_3$. Nevertheless, one can see that it is the gap of a one-dimensional Ising model whose qubita are missing at some of the points. The number of these points is $\frac{|C|}{2}$, where $|C|$ is the length of $C_3$ defined as the number of qubits that lie on $C_3$ loop. Note that there are $\frac{|C|}{2}$ distinct blue sub-lattices in between every two adjacent missing points. Therefore Using inequality \ref{ineq}, one can expand $\text{Gap}(\mathcal{A})$, as a summation of $\frac{|C|}{2}$ terms, where each term is equal to the minimum of the decay rate for each of these $\frac{|C|}{2}$ Ising models that are defined on one of the $\frac{|C|}{2}$ aforementioned sub-lattices. Because the minimum of the decay rate for these $\frac{|C|}{2}$ Ising models are equal to each other, therefore one obtains:

\bea \label{58}
\text{Gap}(\mathcal{A}) &\geq& \frac{|C|}{2} \min_{\Gamma_z}{ \bigg(\sum_{i \in C^1_{b}, i\notin Z_3} \la [\sigma_{x,i}, \mathcal{B}_p^z],[\sigma_{x,i}, \mathcal{B}_p^z]\ra\bigg)} \\ \nonumber
 & \geq & 0
 \eea
 where, by $C_b^1$ we mean that the bath is applying on qubits that lie on one of these $\frac{|C|}{2}$ blue sub-lattices. The reason that the lower bound for the gap of one of these Ising models is $0$, is that one can find an observable belonging to $\mathcal{B}_p^z$, such that the support of this observable do not have any point in common with the vertices of one of these $\frac{|C|}{2}$ blue sub-lattices.

Therefore from Eqs. (\ref{56})-(\ref{58}), one finds the following lower bound for the gap of $\mathcal{L}_{x,C_{b}}^3$ in the present sector:
\be \label{g3}
\mathcal{L}_{x,C_{b}}^3=\mathcal{L}_{x,C_{g}}^3 \geq 0.
\ee
Using Eqs. (\ref{eq:G3min})-(\ref{g3}), the minimum of the decay rate for this sector is given by:
\be
G_3 \geq \mathcal{L}_{x,C_{r}}^3+\mathcal{L}_{x,C_{b}}^3+\mathcal{L}_{x,C_{g}}^3 \geq e^{- \beta \Delta} h(\Delta)  \text{Gap}_{\text{Ising}}.
\ee

\subsubsection{\textbf{The sector of ${\Gamma_z}=Z_1Z_3\mathcal{B}_p^z$ or ${\Gamma_z}=Z_1Z_4\mathcal{B}_p^z$ or ${\Gamma_z}=Z_2Z_3\mathcal{B}_p^z$ or ${\Gamma_z}=Z_2 Z_4\mathcal{B}_p^z$  }}
The minimum of the decay rate in this sector is shown by $G_4$. Using the inequality \ref{ineq} and the notation $\mathcal{L}_{C_{b,g,r}}^4 = \sum_{i \in C_{b,g,r}}\la [\sigma_{x,i},Z_1Z_3 \mathcal{B}_p^z],[\sigma_{x,i},Z_1Z_3 \mathcal{B}_p^z]\ra$ we obtain:
\be
G_4  \geq e^{- \beta \Delta} h(\Delta)\bigg( \min_{\Gamma_z} (\mathcal{L}_{x,C_b}^4)+ \min_{\Gamma_z}( \mathcal{L}_{x,C_g}^4)+  \min_{\Gamma_z}(\mathcal{L}_{x,C_r}^4) \bigg).
\ee

Now we calculate the terms $\mathcal{L}_{x,C_b}^4, \mathcal{L}_{x,C_g}^4$ and $\mathcal{L}_{x,C_r}^4$ separately as follows:
\bea
\mathcal{L}_{x,C_b}^4&=& \sum_{i \in C_b} \la [\sigma_{x,i},Z_1Z_3\mathcal{B}_p^z],[\sigma_{x,i},Z_1Z_3 \mathcal{B}_p^z]\ra \\ \nonumber
&=& \sum_{i \in C_b}\la [\sigma_{x,i}, Z_3\mathcal{B}_p^z],[\sigma_{x,i}, Z_3\mathcal{B}_p^z]\ra,
\eea
where the second line is the consequence of commutativity of $Z_1$ with all of $\sigma_{x,i}, i \in C_b$.  In addition, as we proved in the previous case, a lower bound of this quantity is given by equation \ref{g3} and by symmetry, it is equal to the lower bound of $\mathcal{L}_{x,C_g}^4$. The only term that remains to be obtained is $\mathcal{L}_{x,C_r}^4$. One finds that:

\bea
\mathcal{L}_{x,C_r}^4&=& \sum_{i \in C_r} \la [\sigma_{x,i},Z_1Z_3\mathcal{B}_p^z],[\sigma_{x,i},Z_1Z_3 \mathcal{B}_p^z]\ra \\ \nonumber
&=& \sum_{i \in C_r}\la [\sigma_{x,i}, Z_1\mathcal{B}_p^z],[\sigma_{x,i}, Z_1\mathcal{B}_p^z]\ra \\ \nonumber
&=&\sum_{i \in C_r}\la [\sigma_{x,i}, Z_2 Z_g\mathcal{B}_p^z],[\sigma_{x,i}, Z_2Z_g\mathcal{B}_p^z]\ra \\ \nonumber
&=&\sum_{i \in C_r}\la [\sigma_{x,i},\mathcal{B}_p^z],[\sigma_{x,i},\mathcal{B}_p^z]\ra,
\eea
which is what we discussed in the first case and is equal to the gap of Ising model. Therefore, a lower bound of $G_4$ is given by:

\be
G_4 \geq  e^{- \beta \Delta} h(\Delta)  \text{Gap}_{\text{Ising}} .
\ee
To obtain the gap of $\mathcal{L}$ it is enough to consider only the above sectors of the observables, since each of the other sectors is equivalent to one of the four mentioned cases, and this is straightforward to check. Therefore, the minimum of the decay rate in all of the sectors can be obtained by doing minimization only over these four sectors. Finally, we arrive at our key theorem:

\begin{theorem*}
The gap of Lindblad super-operator for Topological Color Code due to the coupling to a thermal bath is given by:
\be \label{theorem}
(\text{Gap} (\mathcal{L}))_{\text{TCC}} \geq e^{- \beta \Delta} h(\Delta) (\text{Gap} (\mathcal{L}))_{\text{Ising}}.
\ee
\end{theorem*}

\section{Conclusion}\label{conclusion}
In this work, we have studied thermal stability of the Topological Color Code in presence of a thermal bath of the form \ref{interaction}. To this end, we have studied the Lindblad evolution of the observables in the Hisenberg picture and their auto-correlation functions. The observables that we studied commute with the Hamiltonian in order to be regarded as  logical operators acting on the code space. We obtain a lower bound of the decay rate of these observables as follows:
\be
\la X,X(t) \ra \leq e^{- \epsilon t}  \la X,X \ra,
\ee
where
\be
 \epsilon \geq  e^{- \beta \Delta} h(\Delta) (\text{Gap} (\mathcal{L}))_{\text{Ising}},
\ee
$(\text{Gap} (\mathcal{L}))_{\text{Ising}}$ turns out to be a constant independent of the system size \cite{Horodecki2} and $\Delta$ is the gap of the TCC model which is equal to $6J$. Our result means that the auto-correlation function of the observables decreases exponentially in time faster than an exponential with a constant decay rate independent of the system size, i.e. by increasing the system size one cannot decrease the decay rate to make the memory stable. Thus, in a finite time the auto-correlation function becomes very small and the encoded information will be lost. Therefore, we can conclude that the Topological Color Code is unstable against thermal fluctuations from the bath at finite temperature, even though it is stable at $T=0$ against local quantum perturbations.

 Although the conclusion about the thermal instability of the Color Code is the same as that of the Kitaev code, notice however that the derivation of this new result is very different in the case of Color Code from the case of Kitaev model. For example in Kitaev model excitations appear in pairs as apposed to Color Code, in which excitations do not appear necessarily in pairs. Moreover, in the Kitaev model, to have all possible excitations, one should apply tensor products of $\sigma_x$'s ($\sigma_z$'s) over qubits belonging to the subsets of the snake (comb), on the ground states (for further details see \cite{Horodecki2}). This is in contrast to the Color Code where, to have all possible excitations, one should apply a completely different procedure using the concepts of open strings and the branching points as defined in Sec. \ref{Exc}. As explained in Sec. \ref{Obs}, the generators of the observables for any stabilizer Hamiltonian are the stabilizers which are in the Hamiltonian as well as the generators needed for creating all kind of excitations. Apart from the difference of the stabilizers in the two models, because the generators needed for creating all kinds of excitations in the case of Color Code are different from that of Kitaev, one can conclude that the generators needed to have all observables in the case of Color Code are different from the Kitaev. The last distinctive point is that the process of creation of the excitations caused by the external bath in the TCC can be mapped to the corresponding process in-homogeneous one-dimensional Ising model, in contrast to the case of the Kitaev model which can be mapped to the one-dimensional homogeneous Ising model \cite{Horodecki2}.

The impact of these results goes beyond the field of quantum computation and 
extends to the new emerging field of topological orders in condensed matter system
(strongly correlated spins). In fact, it is known that two-body Hamiltonians in 2D lattices
can give rise to Topological Color Codes in the low-energy sector for certain regimes
of the couplings \cite{color2body1,color2body2}. These topological orders are expected to suffer from thermal instabilities
as well.

\section{Acknowledgement}
The author wishes to thank M.A. Martin-Delgado and Markus M\"uller for fruitful discussions. This work has been done during the author's stay at Complutense University of Madrid and the author wishes to thank the Department of Physics of Complutense University for hospitality and partial financial support. The author also would like to thank V. Karimipour for introducing her to M.A. Martin-Delgado's research group and also reading the manuscript. The author also wishes to thank National Elites Foundation of Iran for partial financial support.

\appendix
\section{Negativity of Lindblad super-operator}\label{appendix a}

\begin{lemma*}
The Lindblad super-operator is negative which means that:
\be
-\la X, \mathcal{L}(X)\ra \geq 0, \qquad  \forall  X.
\ee
\end{lemma*}

\begin{proof}
In order to prove the positivity of $- \la X, \mathcal{L}(X)\ra$ we use the definition of $\mathcal{L}(X)$:
\bea
-\mathcal{L}(X)&=& - \sum_\alpha \sum_{\omega \geq 0} \mathcal{L}_{\alpha,\omega} (X)\\ \nonumber
&=& \frac{1}{2} \sum_\alpha \sum_{\omega \geq 0} h_\alpha(\omega) \Bigg( S^\dagger_\alpha(\omega) [ S_\alpha(\omega),X]+[X,S^\dagger_\alpha(\omega)] S_\alpha(\omega) \\ \nonumber
&+& e^{- \beta \omega} S_\alpha(\omega) [ S^\dagger_\alpha(\omega),X] +e^{- \beta \omega} [X,S_\alpha(\omega) ] S^\dagger_\alpha(\omega)\Bigg). \nonumber
\eea
Since $\frac{1}{2} h_\alpha(\omega)$'s are positive, we prove the positivity of each term $- \la X, \mathcal{L}_{\alpha,\omega}(X)\ra$:
\bea
 - \la X, \mathcal{L}_{\alpha,\omega}(X)\ra&=& \la X, S^\dagger_\alpha(\omega) [S_\alpha(\omega),X]+[X,S^\dagger_\alpha(\omega)] S_\alpha(\omega) \ra \\ \nonumber
 &+& e^{- \beta \omega} \la X, S_\alpha(\omega) [ S^\dagger_\alpha(\omega),X] + [X,S_\alpha(\omega) ] S^\dagger_\alpha(\omega) \ra.
\eea
Expanding the above equation and using the relation $\rho_\beta S_\alpha(\omega) = e^{\beta \omega} S_\alpha(\omega) \rho_\beta$, one can find that:
\be \label{positive}
- \la X, \mathcal{L}_{\alpha,\omega}(X)\ra= \la [S_\alpha(\omega),X],[S_\alpha(\omega),X]\ra+ e^{-\beta \omega} \la [S^\dagger_\alpha(\omega),X],[S^\dagger_\alpha(\omega),X]\ra.
\ee
In order to prove the positivity of $\la [S_\alpha(\omega),X],[S_\alpha(\omega),X]\ra$ one needs to use the explicit form of $S_\alpha(\omega)$, which for $x$ type errors is as follows:

\begin{figure}[t]
\begin{center}
\includegraphics[scale=0.32]{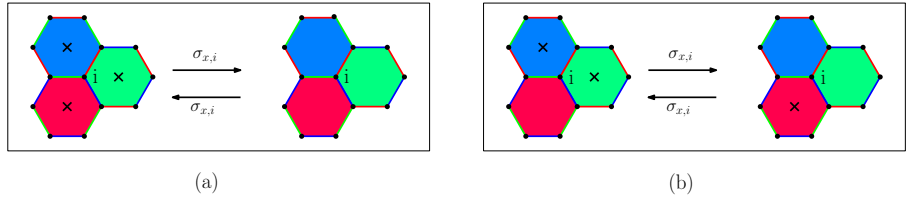}
\caption{(Color Online): Two possibilities of creation and annihilation of excitations in TCC due to existence of a thermal bath. a) Creation and annihilation of three excitations. b) Creation of one excitation and annihilation of two, and vice versa.}\label{bath}
\end{center}
\end{figure}

 \begin{itemize}
 \item \textbf{Annihilation of three excitations}: If there are three excitations in the three plaquettes that have this qubit in common, $\sigma_{x,i}$ annihilates all of them (figure \ref{bath}a). Thus, the Fourier transform of $\sigma_{x,i}$ is given by:
\be
S_{x,i} (\omega=6J)= \Pi_{-3J} \sigma_{x,i} \Pi_{3J}= \frac{1}{8}\sigma_{x,i} (1-B^z_p) (1-B^z_{p^\prime}) (1-B^z_{p^{\prime\prime}}).
\ee
Here, $p$, $p^\prime$ and $p^{\prime \prime}$ are the three plaquettes that have the $i^{th}$ qubit in common and  $\Pi_{3J}$ denotes the projector onto the sub-space with three excitations and $\Pi_{-3J}$ denotes the projector onto the sub-space with no excitation in these three plaquettes.
 \item \textbf{Creation of one excitation and annihilation of two}: If there are two excitations in two of the plaquettes that have this qubit in common, $\sigma_{x,i}$ annihilates them and creates one excitation in the other plaquette (figure \ref{bath}b). Thus, the Fourier transform of $\sigma_{x,i}$ is given by:
 \bea \nonumber
&&S_{x,i} (\omega=2J)= \Pi_{-J} \sigma_{x,i} \Pi_J= \frac{1}{8}\sigma_{x,i}\bigg( (1-B^z_p) (1-B^z_{p^\prime}) (1+B^z_{p^{\prime\prime}}) \\
&&+(1-B^z_p) (1+B^z_{p^\prime}) (1-B^z_{p^{\prime\prime}})
+(1+B^z_p) (1-B^z_{p^\prime})(1-B^z_{p^{\prime\prime}}) \bigg).
\eea
Here, $\Pi_{J}$ denotes the projector onto the sub-space with two excitations and $\Pi_{-J}$ denotes the projector onto the sub-space with one excitation.
\end{itemize}

Therefore, we have:
\bea
&&S_{x,i}(6J)=\sigma_{x,i} \Pi_{3J}, \qquad \qquad S^\dagger_{x,i}(6J)=\sigma_{x,i} \Pi_{-3J},\\ \nonumber
&&S_{x,i}(2J)=\sigma_{x,i} \Pi_{J}, \qquad \qquad S^\dagger_{x,i}(2J)=\sigma_{x,i} \Pi_{-J}.
\eea
 All of the above arguments can be done in a similar fashion for the $z$ type error by substituting $\sigma_{z,i}$ for $\sigma_{x,i}$ and $B^x_P$ for $B^z_p$.\\
 
It is sufficient to prove the positivity of $\la [S_\alpha(\omega),X],[S_\alpha(\omega),X]\ra$ for a specific $\omega$, say $6J$. For the other $\omega$'s the procedure is the same. For this case $\la [S_\alpha(\omega),X],[S_\alpha(\omega),X]\ra$ is equal to the following:
\be
\text{tr} \bigg( \rho_\beta \Pi_{3J} [X, \sigma_{x,i}] [\sigma_{x,i},X] \Pi_{3J} \bigg)=\text{tr} \bigg( \Pi_{3J} \rho_\beta \underbrace{[X, \sigma_{x,i}]}_{A^\dagger} \underbrace{[\sigma_{x,i},X] }_A \bigg).
\ee

Since $\Pi_{3J}$ is a projector onto a sub-space of the system's Hilbert space it can be written as:
\be
\Pi_{3J} = \sum_m |\Psi_m\ra \la \Psi_m|.
\ee
The thermal state ($\rho_\beta$) also can be written as a mixture of eigenstates of the Hamiltonian, i.e.:
\be
\rho_\beta= \sum_k  \lambda_k |\Psi_k \ra \la \Psi_k|.
\ee
Therefore, we have:
\be
 \Pi_{3J} \rho_\beta=\sum_m \lambda_m|\Psi_m\ra \la \Psi_m|.
 \ee
If we diagonalize matrix $A$ and expand the eigenstates of the Hamiltonian as a superposition of the eigenstates of $A$, i.e.
\be
|\Psi_m\ra=\sum_a \psi_{m,a} |\Psi_{m,a} \ra,
\ee
we will end in the following:
\be
\text{tr} \bigg( \rho_\beta \Pi_{3J} [X, \sigma_{x,i}] [\sigma_{x,i},X] \Pi_{3j} \bigg)= \sum_{m,a} \lambda_m |\psi_{m,a}|^2,
\ee
which is clearly positive.\\

These arguments are not particular for TCC. For the other models $\omega$'s and $\Pi$'s are different, nevertheless, the procedure of the proof is the same as above.
\end{proof}

{}


\begin{thebibliography}{}
\bibitem{foot1}
We denote the red plaquettes by R, the green plaquettes by G and the blue plaquettes by B
\bibitem{foot2}
A stabilizer is an operator which acts as identity in the code space.
\bibitem{foot3}
One should note that the red links have 3 distinct orientations as shown in figure \ref{lattice}, but only the links with one specific orientation are sufficient here.
\bibitem{foot4}
In order for the genarators to commute with each other we take two distinct branching points for the two types $x$ and $z$. The branching points are shown as $b_x$ and $b_z$ in figure \ref{excitations}b
\bibitem{foot5}
$Z_i$ and $X_i$ operators are defined on the non-trivial loops shown in figure \ref{lattice}b.
\end{thebibliography}
\end{document}